\documentclass[UKenglish,cleveref,autoref,thm-restate,runningheads]{llncs}
\usepackage{hyperref}
\usepackage{amsmath}
\usepackage[x11names,table]{xcolor}
\usepackage{todonotes} 
\usepackage[T1]{fontenc}
\usepackage{txfonts} % \medbullet, \medcirc
\usepackage{CJKutf8}
\usepackage{multirow} 
\usepackage{wrapfig}
\usepackage{enumitem}
\usepackage{environ}
\usepackage{wrapfig}
\usepackage{subcaption}
\usepackage{tabularx}

\newcommand{\graph}[1]{\mathsf{#1}}

\newcommand{\CubeGraph}[1]{\graph{H}_{#1}} % Hypercube

\newcommand{\NCGraph}[1]{\graph{S}^{\times}_{#1}}

\newcommand{\TUPLE}{{\textsc{2TupleGC}}}
\newcommand{\TUPLEC}{{\textsc{2TupleGC'}}}
\newcommand{\BITSTRING}{{\textsc{BitstringGC}}}
\newcommand{\GRID}{{\textsc{GridHamPath}}}
\newcommand{\COMBSWAP}{{\textsc{CombSwapGC}}}
\newcommand{\COMBTRANS}{{\textsc{CombTransGC}}}
\newcommand{\COMBCOMP}{{\textsc{CombCompGC}}}
\newcommand{\COMBREV}{{\textsc{CombRevGC}}}
\newcommand{\PERMSWAP}{{\textsc{PermSwapGC}}}
\newcommand{\PERMTRANS}{{\textsc{PermTransGC}}}
\newcommand{\PERMREV}{{\textsc{PermRevGC}}}
\newcommand{\PERMROT}{{\textsc{PermRotGC}}}
\newcommand{\PERMJUMP}{{\textsc{PermJumpGC}}}

\newcommand{\STEDGEEXG}{{\textsc{SpanningTreeGC}}}
\newcommand{\PMALTCYCLE}{{\textsc{PerfectMatchingGC}}}
\newcommand{\NCREF}{{\textsc{NCSetPartRefGC}}}
\newcommand{\SPREF}{{\textsc{SetPartRefGC}}}

\newcommand{\objSet}[1]{\mathbb{#1}}
\newcommand{\BINARY}[1]{\objSet{B}_{#1}}
\newcommand{\COMBOS}[2]{\objSet{B}_{#1}^{#2}}

\newcommand{\PERMUTE}[1]{\objSet{P}_{#1}}  
\newcommand{\PEAKLESSPERM}[1]{\objSet{P}^{\lor}_{#1}}

\newcommand{\NC}[1]{\objSet{S}^{\times}_{#1}}

\newcommand{\SPANNING}[1]{\objSet{ST}_{#1}}
\newcommand{\PMATCHING}[1]{\objSet{PM}_{#1}}
\usepackage{todonotes}

\newcommand{\cO}{\mathcal{O}}

\newcommand{\NN}{\mathbb{N}}

\newcommand{\darkBlue}[1]{{\color{Blue4}#1}}

\newcommand{\darkRed}[1]{{\color{Red4}#1}}
\newcommand{\darkGreen}[1]{{\color{Green4}#1}}
\newcommand{\padding}[1]{\darkRed{#1}}

\newcommand{\bits}[1]{\darkBlue{#1}}

\makeatletter
\newcommand{\problemtitle}[1]{\gdef\@problemtitle{#1}}% Store problem title
\newcommand{\probleminput}[1]{\gdef\@probleminput{#1}}% Store problem input
\newcommand{\problemquestion}[1]{\gdef\@problemquestion{#1}}% Store problem question
\NewEnviron{problem}{
  \problemtitle{}\probleminput{}\problemquestion{}% Default input is empty
  \BODY% Parse input
  \par\addvspace{.5\baselineskip}
  \noindent
  \begin{tabularx}{\textwidth}{@{\hspace{\parindent}} l X c}
    \multicolumn{2}{@{\hspace{\parindent}}l}{\underline{\@problemtitle}} \\% Title
    \textbf{Input:} & \@probleminput \\% Input
    \textbf{Question:} & \@problemquestion% Question
  \end{tabularx}
  \par\addvspace{.5\baselineskip}
}
\makeatother 

\graphicspath{{./binary/}{./catalan/}{./combinations/}{./doodles/}{./graphics/}{./graphs/}{./grid/}{./hashi/}{./necklaces/}{./noncrossing/}{./permutations}{./polyominoes/}{./reductions/}{./set/}{./spanning/}{./tilings/}{./yesno}}%helpful if your graphic files are in another directory

\begin{document}
\title{On the Hardness of Gray Code Problems for Combinatorial Objects
% \thanks{Supported by organization x.}
}

% \author{Arturo Merino\inst{1}\orcidID{0000-0002-1728-6936} \and
% Namrata\inst{2}\orcidID{0000-0002-6582-4196} \and
% Aaron Williams\inst{3}\orcidID{0000-0001-6816-4368}}

\author{Arturo Merino\inst{1} \and
Namrata\inst{2} \and
Aaron Williams\inst{3}}
\authorrunning{A. Merino, Namrata, A. Williams} 

\institute{University of Saarland and Max Planck Institute for Informatics, Germany \\
\email{merino@cs.uni-saarland.de}\and
University of Warwick, Department of Computer Science, England\\
\email{namrata@warwick.ac.uk} \and
Williams College, Department of Computer Science, United States\\
\email{aaron.williams@williams.edu}
}
\maketitle   
\begin{abstract}
Can a list of binary strings be ordered so that consecutive strings differ in a single bit?
Can a list of permutations be ordered so that consecutive permutations differ by a swap?
Can a list of non-crossing set partitions be ordered so that consecutive partitions differ by refinement?
These are examples of \emph{Gray coding problems}:
Can a list of combinatorial objects (of a particular type and size) be ordered so that consecutive objects differ by a flip (of a particular type)?
For example, $000,001,010,100$ is a no instance of the first question, while $1234,1324,1243$ is a yes instance of the second question due to the order $12\overline{43}, 1\overline{23}4, 1324$.
We prove that a variety of Gray coding problems are NP-complete using a new tool we call a \emph{Gray code reduction}. 
\end{abstract}

\section{Introduction}
\label{sec:intro}

In a 1947 patent application, Bell Labs engineer Frank Gray devised an order of the~$2^n$ binary strings of length $n$ in which consecutive strings differ by flipping a single bit (i.e., they have Hamming distance one) \cite{gray1953pulse}. 
He referred to the order as \emph{reflected binary code} due to its recursive structure. 
Although the order had previously been observed by others, including another Bell Labs engineer George R. Stibitz \cite{stibitz1943binary}, the order became known as the \emph{binary reflected Gray code (BRGC)}, or simply, the \emph{Gray code}.
% The same order is also implicit in earlier puzzles including \emph{baguenaudier} (``time-waster'') and the \emph{Towers of Hanoi}.

While Bell Labs was able to solve their ordering problem several times, similar pursuits are often quite challenging.
For example, the well-studied \emph{middle levels conjecture} \cite{buck1984gray} asked if the same type of ordering exists for the binary strings of length $2k+1$ with either $k$ or $k+1$ copies of $1$.
Knuth gave this conjecture a difficulty rating of 49/50 \cite{knuth2011art} before it was settled in the affirmative by Mütze \cite{mutze2016proof}, with subsequent work simplifying \cite{mutze2023book}, specializing \cite{merino2022combinatorial}, and generalizing \cite{gregor2023central} \cite{merino2023kneser} the result.

Centuries earlier, bell-ringers developed an order of the $n!$ permutations of $[n] = \{1,2,\ldots,n\}$ (viewed as strings) where consecutive permutations differ by a \emph{swap} (or \emph{adjacent-transposition}) meaning that two neighboring symbols are exchanged \cite{duckworth1668tintinnalogia}.
\emph{Plain changes} was rediscovered independently by Johnson \cite{johnson1963generation}, Trotter \cite{trotter1962algorithm}, and Steinhaus \cite{steinhaus1979one} in the 1960s for its use in the efficient generation of permutations by computer.

In general, when presented with a combinatorial object and a flip operation, one may ask for an order in which successive objects differ by a flip.
Suitable orders are sometimes referred to as \emph{minimal change orders} or \emph{combinatorial Gray codes}.
Academic surveys have been written by Savage \cite{savage1997survey} and more recently M\"{u}tze \cite{mutze2022combinatorial}, with Ruskey \cite{ruskey2003combinatorial} and Knuth \cite{knuth2011art} devoting extensive textbook coverage to the subject.
Despite the long history of the subject, there are still natural Gray code questions that haven't been answered or even posed.
For example, in Section \ref{sec:reductions} we'll consider such a question involving non-crossing set partitions.

\subsection{Gray Codes and Computational Complexity}
\label{sec:intro_complexity}

When Gray codes are mixed with computational complexity the focus is typically on \emph{generation problems}:
How efficiently can a particular order can be generated?
For example, Ehrlich's well-known paper \cite{ehrlich1973loopless} provides \emph{loopless algorithms} for the binary reflected Gray code and plain changes using the \emph{shared object model}.
In other words, one instance of the object is shared between the generation algorithm and the application, and it is modified in worst-case $O(1)$-time to create the next instance.
More recent work has focused on limiting generation algorithms to constant \emph{additional memory}~\cite{stevens2014coolest} \cite{liptak2023constant}.

We instead show that there are computationally hard existence problems that underlie the problems solved by Bell Labs engineers, bell-ringers, and many others throughout history.
More specifically, we consider \emph{existence problems} like the following: 
\begin{enumerate}
    \item[{\bfseries Q1}] Can a list of binary strings be ordered so that consecutive strings differ by a~bitflip?
    \item[{{\bfseries Q2}}] Can a list of permutations of be ordered so that consecutive strings differ by a~swap?
\end{enumerate}
For example, $000,001,010,100$ is a no instance of {\bfseries Q1}, while $1234,1324,1243$ is a yes instance of {\bfseries Q2} due to the order $12\overline{43}, 1\overline{23}4, 1324$.
Note that in these decision problems the type of object and flip operation is fixed, and the input is the list of objects under consideration.
To be clear, each object in the list is provided as part of the input, so the size of the input increases along with the number of objects in the list.%
\footnote{Conceptually, the input could be described as a subset of the objects.
However, subsets of an $n$-set are often encoded as $n$-bit incidence vectors, and we want to avoid this misinterpretation.}

We refer to these existence problems as \emph{Gray coding problems}, with the connotation that we are trying to \emph{do} something to the list of objects. %, like graph color\emph{ing} problems.
We consider classic combinatorial objects including binary strings, permutations, combinations, (non-crossing) set partitions, and graphs.
In each case, we identify at least one flip operation for which the Gray coding problem is NP-complete (including {\bfseries Q1} and~{\bfseries Q2}).

\subsection{Outline}
\label{sec:intro_outline}

In Section \ref{sec:first}--\ref{sec:second} we establish that two specific Gray coding problems are NP-complete.
Then Section \ref{sec:reductions} introduces our notion of a Gray code reduction.
Sections \ref{sec:combos}--\ref{sec:graphs} use these reductions to obtain additional hardness results for a variety of combinatorial objects.
% our framework with results summarized in Table~\ref{tab:objects}.
% These application sections are designed to quickly cover a wide variety of results, so some of the proofs omit polynomial-time and NP membership arguments.
% We also draw the hypercube $\CubeGraph{3}$ using layers starting in Section \ref{sec:combos}.
% We also note that every conclusion of NP-completeness is strong NP-completeness, due to the strong NP-completeness of Theorem \ref{thm:gridGraphPath}.
Final remarks are contained in Section \ref{sec:final}.

\section{A First NP-Complete Problem}
\label{sec:first}

In this section, we discuss a first Gray coding problem that is NP-complete.

\subsection{$2$-Tuple Gray Codes}
\label{sec:first_2tuple}

% \todo{Aaron: perhaps introduce \emph{standard form} to mean includes $0$.}

In the \emph{2-tuple Gray coding problem}, we are given some integer 2-tuples, and we want to decide if we can order the 2-tuples such that consecutive 2-tuples differ only in $\pm 1$ in one of the coordinates.
More formally, we use $\PERMUTE{m}$ to denote the permutations of length $m$ and have the following problem. 

\begin{problem}
    \problemtitle{\TUPLE}
    \probleminput{A list $L$ of $m$ integer 2-tuples $(a_1,b_1),\dots, (a_m,b_m) \in \NN\times \NN$.}
    \problemquestion{Is there a $\pm 1$ Gray code for $L$? 
 In other words, is there a permutation $\pi \in \PERMUTE{m}$ such that $|a_{\pi(i)}-a_{\pi(i+1)}|+|b_{\pi(i)}-b_{\pi(i+1)}|=1$ for every $i \in [m-1]$?}
\end{problem}

For technical reasons, we also consider a version of $\TUPLE$ where the integers have no gaps between them. 
We say that a list of integer 2-tuples $(a_1,b_1),\dots, (a_m,b_m)$ is \emph{continuous} if the set of values that $a_i$ and $b_i$ take for $i \in [m]$ are consecutive integers starting from $1$; i.e., $\{a_i \mid i \in [m]\} = [\max_{i\in [m]} a_i]$ and $\{b_i \mid i \in [m]\}=[\max_{i \in I} b_i]$.

\begin{problem}
    \problemtitle{\TUPLEC}
    \probleminput{A continuous list $L$ of $m$ integer 2-tuples $(a_1,b_1),\dots, (a_m,b_m) \in \NN\times \NN$.}
    \problemquestion{Is $\TUPLE(L)$ true?}
\end{problem}

\subsection{Grid Graph Hamiltonicity}
\label{sec:first_grid}

Hamilton Path problems have been central to evolution of computational complexity, dating back to Karp's initial list of 21 NP-complete problems \cite{karp1972reducibility}.

A \emph{grid graph} is a graph, where the vertex set is given by some integer 2-tuples $L = \{(a_1,b_1),\dots, (a_m,b_m)\}$ and there are edges between all pairs of 2-tuples that differ by 1 on a single coordinate.
Since the grid graph is completely defined by the integer 2-tuples, we denote by $\mathop{grid}(L)$ the unique grid graph that has $L$ as vertices.

Of particular relevance to us, is the powerful sharpening of this result by Itai, Papadimitriou, and Szwarcfiter which shows hardness for Hamiltonian paths problems on grid graphs \cite{itai1982hamilton}.
More formally, the following problem is hard.

\begin{problem}
    \problemtitle{\GRID}
    \probleminput{A list $L$ of $m$ integer 2-tuples $(a_1,b_1),\dots, (a_m,b_m) \in \NN\times \NN$.}
    \problemquestion{Is there a Hamilton path in $\mathop{grid}(L)$?}
\end{problem}

% \aaron{make a (footnote?) comment that the input of our problem is 2-tuples, whereas the original problem in \cite{itai1982hamilton} did not clarify what the input is.}

\begin{theorem}[\cite{itai1982hamilton}] \label{thm:gridGraphPath}
$\GRID$ is NP-complete.
\end{theorem}

\subsection{Hardness Results}
\label{sec:first_results}

The problem $\GRID$ can be easily translated into an equivalent $\TUPLE$ problem.
In fact, they are essentially the same problem.

\begin{corollary}
    $\TUPLE$ is NP-complete.
\end{corollary}
\begin{proof}
    We note that both problems have the same input, and $L=(a_1,b_1),\dots,(a_m,b_m)$ is a Hamilton path of $\mathop{grid}(L)$ if and only if $L$ is a $\pm 1$ Gray code.
    \qed
\end{proof}
    
Furthermore, $\TUPLE$ is hard even when we restrict the input to be continuous.
Intuitively, non-continuous inputs give rise to disconnected grid graphs or can be translated to a continuous instance.
We have the following theorem.

\begin{theorem}
    $\TUPLEC$ is NP-complete.
\end{theorem}

\begin{proof}
    It is clear that the problem is in NP, as $\pi \in \PERMUTE{m}$ is a polynomially checkable certificate. 
    
    We reduce from $\TUPLE$.
    Let $L = (a_1,b_1),\dots, (a_m,b_m) \in \NN \times \NN$ be an instance of $\TUPLE$.
    We only need to deal with the case of $L$ being non-continuous, as otherwise we simply map $L$ to itself.
    Let $A=\{a_i \mid i \in I\}$ and $B=\{b_i \mid i \in I\}$.
    If $L$ is non-continuous, then either 
    (1) there exists a partition of $[m]$ into $I$ and $J$, and $\alpha \in \NN$ such that $a_i < \alpha < a_{j}$ for every $i \in I$ and $j \in J$,
    (2) there exists a partition of $[m]$ into $I$ and $J$, and $\alpha \in \NN$ such that $b_i < \alpha < b_{j}$ for every $i \in I$ and $j \in J$,
    (3) $A$ and $B$ are the discrete intervals $A = \{\min A, \dots, ,\max A\}$ and $B= \{\min B, \dots, \max B\}$.
    Furthermore, we can decide if we are in case 1, 2 or 3 in time $\cO(m\log m)$ by sorting $A$ and $B$.

    If (1) or (2) holds, then for every $i \in I$ and $j \in J$ we have that $|a_i-a_j|+|b_i-b_j|\geq 2,$ which implies that $L$ is a no-instance.
    If (3) holds, we map $L$ to $L':= (a'_1,b'_1),\dots, (a_m',b_m')$ by shifting the instance so that the minimum among the first and second coordinates is one; i.e.,
    \[(a'_1,b'_1),\dots, (a_m',b_m') = (1+a_1-\min A,1+b_1-\min B),\dots, (1+a_m- \min A,1+b_m - \min B).\]
    Note that the encoding size of $L'$ is at most the encoding size of $L$ and that the mapping can be computed in polynomial time.
    Furthermore, for every $i,j \in [m]$ we have that
    \begin{align*} 
        |a'_i-a'_j|+|b'_i-b'_j| & = |1+a_i- \min A-(1+a_j- \min A)| + |1+b_i- \min B-(1+b_j- \min B)| \\
        & = |a_i-a_j|+|b_i-b_j|,
    \end{align*}
    so $L$ is a yes-instance if and only if $L'$ is a yes-instance. 
    This concludes the proof.
    \qed
\end{proof}

\subsection{Application: Swap Gray Codes for Permutations}
\label{sec:first_swap}

Here we show that $\TUPLEC$ can be used as a source problem for establishing the hardness of other Gray coding problems.
Consider the following problem.
\begin{problem}
    \problemtitle{\PERMSWAP}
    \probleminput{A list of $m$ permutations of length $n$, $\tau_1,\dots, \tau_m \in \PERMUTE{n}$.}
    \problemquestion{Is there a permutation $\pi \in \PERMUTE{m}$ such that $\tau_{\pi(i)}$ and $\tau_{\pi(i+1)}$ differ in an adjacent transposition for every $i \in [m-1]$?}
\end{problem}

\begin{theorem}
    \label{thm:permSwap}
    $\PERMSWAP$ is NP-complete.
\end{theorem}
\begin{proof}
    It is clear that the problem is in NP. 
    To see hardness, we reduce from $\TUPLEC$.

    Let $L=(a_1,b_1),\dots,(a_m,b_m) \in \NN\times \NN$ be a list of continuous 2-tuples. 
    Let $a = \max\limits_{i \in [m]} a_i$ and $b=\max\limits_{i \in [m]} b_i$.
    Given a 2-tuple $(x,y) \in [a] \times [b]$, we define a permutation $\tau := \phi(x,y) \in \PERMUTE{a+b+2}$ as the unique permutation of length $a+b+2$ such that $\tau^{-1}(a+b+1) = x$,  $\tau^{-1}(a+b+2) = a+y$, and after the removing symbols $(a+b+1)$ and $(a+b+2)$ from $\tau$, we get the identity permutation $1\cdots (a+b)$.
  
    We map the instance, $L$ to $L' = \phi(a_1,b_1),\dots, (a_m,b_m)$.
    Note that every permutation has encoding size of $(a+b+2)\log (a+b+2) \leq (2m+2)\log (2m+2)$ and that $\phi$ can be implemented in polynomial time.

    Note that the permutations produced can only differ on swaps involving either the symbol $(a+b+1)$ or the symbol $(a+b+2)$. 
    Furthermore, for two permutations produced by $\phi$, say $\tau,\rho \in \phi([a]\times [b])$, they differ if and only if the positions of symbol $(a+b+1)$ are adjacent or the positions of the symbol $(a+b+2)$ are adjacent, but not both; i.e.,
    \begin{equation}
        \label{eq:adjacentperm}
        |\tau^{-1}(a+b+1)-\rho^{-1}(a+b+1)|+|\tau^{-1}(a+b+2)-\rho^{-1}(a+b+2)|=1.
    \end{equation}
    Finally, if $\tau = \phi(x,y)$ and $\rho =\phi(z,w)$, the left side of \eqref{eq:adjacentperm} is $|x-z|+|y-w|$.
    Consequently, $L$ is a yes-instance if and only if $L'$ is a yes-instance, and the theorem follows.
    \qed
\end{proof}

We will see an alternative proof of Theorem~\ref{thm:permSwap} in later subsections.

\begin{figure}
    \centering
    \begin{subfigure}{0.48\textwidth}
       \centering      
       \includegraphics{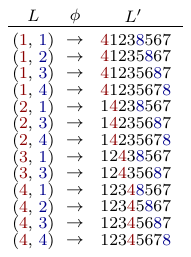}
        \caption{The reduction $\phi$.}
        \label{fig:permswap1}
    \end{subfigure}
    \hfill
    \begin{subfigure}{0.48\textwidth}
        \centering
        \includegraphics{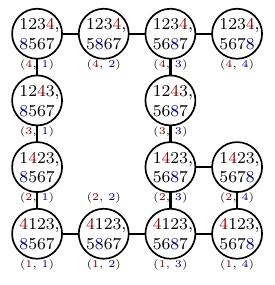}
        \caption{The corresponding grid graph.}
        \label{fig:permswap2}
    \end{subfigure}
    \caption{The reduction used in Theorem~\ref{thm:permSwap}, where $L$ is an instance of $\TUPLEC$ and $L'$ has the corresponding permutations in (a), with the resulting grid graph in (b).}
    \label{fig:enter-label}
\end{figure}
\vspace{-5ex}
\section{A Second Source for NP-Completeness}
\label{sec:second}

While $\TUPLEC$ is a useful source problem, we find it convenient to introduce another NP-complete Gray coding problem for subsequent reductions.
Recall that the binary reflected Gray code lists all $2^n$ bitstrings so that consecutive strings differ in one bit. %, this is arguably the most famous Gray code.
Thus, it is natural to ask which bitstrings have bitflip Gray codes (i.e., {\bfseries Q1} from Section~\ref{sec:intro}).
We'll show that this $\BITSTRING$ problem is hard by a reduction from $\TUPLEC$.

% More formally, we consider the following decision problem.

\begin{problem}
    \problemtitle{\BITSTRING}
    \probleminput{A list of $m$ bitstrings of length $n$, $x_1,\dots, x_m \in \{0,1\}^n$.}
    \problemquestion{Is there a permutation $\pi \in \PERMUTE{m}$ such that $x_{\pi(i)}$ and $x_{\pi(i+1)}$ differ in a bitflip for every $i \in [m-1]$?}
\end{problem}

\begin{theorem}
\label{thm:bitstringsGC}
    $\BITSTRING$ is NP-complete. 
\end{theorem}
\begin{proof}
    It is clear that $\BITSTRING$ is in NP, as the permutation $\pi \in \PERMUTE{m}$ is a polynomially checkable certificate. 
    Thus, it only remains to prove NP-hardness.

    We reduce from $\TUPLEC$. 
    Let $L = (a_1,b_1),\dots, (a_m,b_m) \in \NN \times \NN$ be a list of continuous 2-tuples.
    Let $a = \max\limits_{i \in [k]}$ and $b = \max\limits_{i \in [k]} b_i$. 
    For each tuple $(a_i,b_i)$ we define the bitstring $x_i \in \{0,1\}^{a+b}$ as $x_i := 0^{a_i}1^{a-a_i}0^{b_i}1^{b-b_i}$.
    Since $L$ is continuous, we have $a + b \leq 2m$, and consequently, the mapping of the $\TUPLEC$ instance $L$ to the $\BITSTRING$ instance $L':= x_1,\dots, x_m$ runs in time polynomial in the encoding size of~$L$.

    We now show that $I$ is a yes-instance for $\TUPLEC$ if and only if $I'$ is a yes-instance for $\BITSTRING$.
    Note that for $i,j \in [k]$ the bitstrings $x_i = 0^{a_i}1^{a-a_i}0^{b_i}1^{b-b_i}$ and $x_j = 0^{a_j}1^{a-a_j}0^{b_j}1^{b-b_j}$ differ in a single bitflip if and only if $|a_i-a_j|+|b_i-b_j|=1$. 
    Hence, for every permutation $\pi \in \PERMUTE{m}$ and every $i\in [m-1]$ it holds that $x_{\pi(i)}$ and $x_{\pi(i+1)}$ differ in a single bitflip if and only if $|a_i-a_j|+|b_i-b_j|=1$.
    This concludes the proof. 
    \qed
\end{proof}

Note that the mapping used in Theorems~\ref{thm:bitstringsGC} and~\ref{thm:permSwap} is not polynomial time without continuity.
This can be easily seen, in the non-continuous input 2-tuple $(n,n)$ which needs $\cO(\log n)$ bits to be represented, but it is mapped to the bitstring $0^{2n}$ that needs $\cO(n)$ bits.

A similar reduction idea has been also used to show that solving the Rubik's cube optimally is hard~\cite{demaine2018solving}.

\section{Polynomial-time Gray Code Reductions via Hypercubes}
\label{sec:reductions}

There is a natural graph associated with every Gray code:
represent each object with a vertex, and join two vertices by an edge if their objects differ by a flip.
These graphs are known as ~\emph{flip graphs}, and a Gray code provides a Hamilton path.
For example, the flip graph for bitstrings of length $n$ and bitflips is the $n$-dimensional hypercube or $n$-cube.

% involving various combinatorial objects and flips operations.
Let $Y$ be a type of combinatorial object, and $Y_m$ be those objects of size $m$.
In addition, let $F: Y \rightarrow Y$ be a type of flip operation acting on objects of type $Y$ without changing their size.
In particular, let $F_m: Y_m \rightarrow Y_m$ be the flip operation applied to the objects of size $m$.
To prove that the Gray coding problem on $Y$ and $F$ is hard, we will use a new type of reduction defined below.

\begin{definition} \label{def:HypercubeGrayCodeReduction}
A \emph{polynomial-time Gray code reduction via hypercubes to $Y$ and $F$} is a poly-time computable function~$f : \BINARY{n} \to Y_m$ that maps bitstrings of length~$n$ to an object of type $Y$ and size $m$, such that two binary string $b \in \BINARY{n}$ and $b' \in \BINARY{n}$ differ in a single bit if and only if the corresponding objects $f(b) \in Y_m$ and $f(b') \in Y_m$ differ by a flip of type $F_m$.
\end{definition}

For brevity, we use the term \emph{Gray code reduction} for \emph{polynomial-time Gray code reduction via hypercubes} in the rest of the document. 
An immediate consequence of Definition \ref{def:HypercubeGrayCodeReduction} is the following remark.

\begin{remark} \label{rem:InducedSubgraph}
If there is a Gray code reduction from objects $Y$ and flips $F$, then the flip graph $(Y_m, F_m)$ contains an induced subgraph that is isomorphic to hypercubes of dimension~$n$ inside the flip graph of dimension~$m$ associated with~$Y$. Moreover, we can efficiently find induced subgraphs of the flip graph that are isomorphic to any induced subgraph of the hypercube.    
\end{remark}

%This is due to the fact that the function $f$ can be computed in polynomial-time.
We now present our main theorem for proving that various Gray coding problems are NP-hard.

\begin{theorem} \label{thm:GrayCodeReduction}
    If there is a Gray code reduction $f: \BINARY{n} \to Y_m$ for flips of type $F_m$, then the following Gray coding problem is NP-hard.
    \begin{problem}
    \problemtitle{Gray Coding Problem for Objects $Y_m$ and Flips $F_m$}
    \probleminput{A list $L$ of elemenets in $Y_m$.}
    \problemquestion{Is there a $F_m$ flip Gray code for $L$? 
    }
\end{problem}
\end{theorem}
\begin{proof}
Consider a list of binary strings $B \subseteq \BINARY{n}$, and the associated bitflip Gray code problem $\BITSTRING(L)$.
If there is a Gray code reduction $f: \BINARY{n} \to Y_m$ for flips of type $F_m$, then consider the following list 
\begin{equation} \label{eq:GrayCodeReduction_S}
L := \{f(b) \mid b \in B\}.
\end{equation}
By definition of a Gray code reduction, we know that $b \in \BINARY{n}$ and $b' \in \BINARY{n}$ differ in a single bit if and only if $f(b) \in Y_m$ and $f(b') \in Y_m$ differ by a flip of type $F_m$.
Therefore, $\BITSTRING(L)$ is a yes-instance, if and only if, $GrayCoding(L, F_m)$ is a yes-instance.
Also, note that $L$ can be created in polynomial-time with respect to the size of the original input $B$. %due to the fact that it runs in polynomial-time for each member of $B$ with respect to $n$. 
Since $\BITSTRING$ is NP-hard, we conclude that $GrayCoding(L, F_m)$ is NP-hard.
\qed
\end{proof}

\subsection{Example: (Non-Crossing) Set Partitions by Refinement}
\label{sec:reductions_sp}

To visualize Theorem \ref{thm:GrayCodeReduction}, let's consider a Gray coding problem that has not been posed in the literature.
Let $\NC{n}$ be the set of \emph{non-crossing set partitions of $[n]$}, which are set partitions in which no pair of subsets cross (i.e., if $a$ and $b$ are in one subset and $x$ and $y$ are in another, then $a < x < b < y$ is not true).
Two different set partitions differ by a \emph{refinement} if one can be obtained from the other by splitting a single subset or merging two subsets.
The corresponding flip graph is $\NCGraph{n}$.

\begin{problem}
    \problemtitle{{$\NCREF$}}
    \probleminput{A list $L$ of non-crossing partitions from $\NC{n}$.}
    \problemquestion{Is there a refinement Gray code for $L$? 
    }
\end{problem}

A Gray code reduction from $\BITSTRING$ to $\NCREF$ problem is shown in Figure \ref{fig:GrayCodeReduction}.
In particular, we map binary strings of length $n$ to non-crossing set partitions of $[n+1]$ as follows:
if $b_i = 0$, then $i+1$ is a singleton subset, and otherwise $i+1$ is in the same subset as $1$.
Thus, toggling $b_i$ is equivalent to moving element $i+1$ in or out of $1$'s subset, and this move is a refinement.
As a result, the mapping provides an induced subgraph of $\NCGraph{n+1}$ that is isomorphic to $\CubeGraph{n}$ (as highlighted).
Moreover, $f$ can be computed in polynomial-time for each binary string, so we can efficiently find an induced subgraph of $\NCGraph{n+1}$ that is isomorphic to any induced subgraph of $\CubeGraph{n}$.
Hence, we can conclude that $\NCREF$ is NP-hard.
Careful readers may have noticed that hardness also follows for $\SPREF$ (i.e., $\NCREF$ but on set partitions) since every non-crossing set partition is also a set partition.
% Similarly, the problems $\NCMOVE$ and $\SPMOVE$ (i.e., $\NCREF$ and $\SPREF$ but using element moves instead of refinements) are also NP-hard since the refinements .
Both problems are also clearly in NP, so we have the following theorem.

\begin{theorem} \label{thm:SetPartitionRefinement}
$\NCREF$ and $\SPREF$ are NP-complete.
\end{theorem}

\begin{figure}[h]
    \centering
    \begin{subfigure}[b]{0.32\textwidth}
        \centering       \includegraphics[scale=0.9]{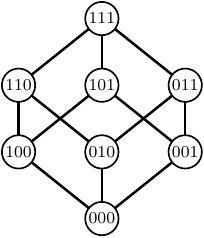}
        \caption{The $3$-cube $\CubeGraph{3}$.}
        \label{fig:GrayCodes_binary}
    \end{subfigure}
    \hfill
    \begin{subfigure}[b]{0.2\textwidth}
        \centering
        \begin{tabular}{@{}c@{\;\;}c@{\;\;}c@{}}
        $\BINARY{3}$ & $f$ & $S \subseteq \NC{4}$ \\ \hline
            000 & $\rightarrow$ & $1|234$ \\
            001 & $\rightarrow$ & $14|23$ \\
            010 & $\rightarrow$ & $13|24$ \\
            011 & $\rightarrow$ & $134|2$ \\
            100 & $\rightarrow$ & $12|34$ \\
            101 & $\rightarrow$ & $124|3$ \\
            110 & $\rightarrow$ & $123|4$ \\
            111 & $\rightarrow$ & $1234$
        \end{tabular} 
        \caption{Mapping.}
    \end{subfigure}
    \hfill
    \begin{subfigure}[b]{0.46\textwidth}
        \centering
        \includegraphics[scale=0.9]{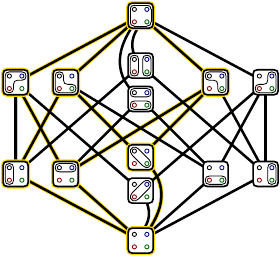}
        \caption{The flip graph $\NCGraph{4}$ with $\NCGraph{4}[S]$ highlighted.}
        \label{fig:GrayCodes_noncrossing}
    \end{subfigure}
    \caption{A Gray code reduction from $\BITSTRING$ to $\NCREF$.
    The one-to-one function $f : \BINARY{n} \to \NC{n+1}$ maps binary strings to non-crossing set partitions in such a way that $b \in \BINARY{n}$ and $b' \in \BINARY{n}$ differ by a bit-flip, if and only if, $f(b) \in \NC{n+1}$ and $f(b') \in \NC{n+1}$ differ by refinement.
    %Therefore, the answer to $\BITSTRING(B)$ for $B \subseteq \BINARY{n}$ is the same as the answer to $\NCREF(S)$ for $S = \{f(b) \mid b \in B\}$. 
    %Moreover, $f(b)$ can be computed in polynomial-time for each $b$ (see Section \ref{sec:reductions}).
    In (c) we use \textcolor{black}{$\mathbf{\circ}$}, \textcolor{red}{$\mathbf{\circ}$}, \textcolor{green}{$\mathbf{\circ}$}, \textcolor{blue}{$\mathbf{\circ}$} for $1,2,3,4$ and the non-singleton subsets are surrounded.
    }
 
    \label{fig:GrayCodeReduction}
\end{figure}
In subsequent sections, we visualize our Gray code reductions by illustrating the induced subgraph isomorphic to a hypercube.
In other words, we illustrate the list $L = \{f(b) \mid b \in \BINARY{n}\}$ (i.e., \eqref{eq:GrayCodeReduction_S} for the full set of binary strings) and argue that the mapping $f$ can be computed in polynomial-time.

For the sake of comparison, the approach taken in Sections~\ref{sec:first_swap} and~\ref{sec:second} can be defined as a \emph{polynomial-time Gray code reduction via grid graph}, where the source problem was~$\TUPLEC$ rather than~$\BITSTRING$.

\section{Combinations}
\label{sec:combos}

A \emph{combination} is a subset of $[n]$ of a fixed size $k$, where $0 \leq k \leq n$.
We denote the set of all combinations of $[n]$ of fixed size $k$ by ~$\COMBOS{n}{k}$.
Commonly, it is represented as a binary string $b_1 b_2 \ldots b_n$ where $b_i := 1$ if $i$ is in the set, otherwise $b_i := 0$. 
The first transposition Gray code for combinations appeared in \cite{tang1973distance}. 
Subsequently, many Gray code results were published in the literature notably~\cite{eades1984algorithm}, \cite{buck1984gray}, and ~\cite{ruskey2009coolest}.

\begin{problem}
    \problemtitle{{$\COMBSWAP / \COMBTRANS / \COMBCOMP/ 
    \COMBREV$}}
    \probleminput{A list $L$ of combinations from $\COMBOS{n}{k}$.}

    \problemquestion{Is there a swap / transposition / substring complement/ substring reversal Gray code for $L$? 
    }
\end{problem}

\begin{theorem} \label{thm:combos}
$\COMBSWAP, \COMBTRANS, \COMBCOMP,  \COMBREV$ are NP-hard.
\end{theorem}
\begin{proof}
We use a Gray code reduction; see Figure~\ref{fig:combos}.
For $\COMBSWAP$, we define the following list of combinations from $\COMBOS{2n}{n}$
\begin{equation} \label{eq:comboSwap}
    L := \{ \bits{c_1 c_2} \ldots \bits{c_{2n-1} c_{2n}} \; | \; 
    \bits{c_{2i-1} c_{2i}} \in \{01,10\} \text{ for } 1 \leq i \leq n \, \},
\end{equation}
where each bit $b_i$ is implemented by the pair $\bits{c_{2i-1} c_{2i}}$.
We map each binary string of length $n$ to a combination of length $2n$ with $n$ many 1s as
\begin{equation} \label{eq:comboSwapTranslate}
    f(b_1 b_2 \cdots b_n) := \darkBlue{c_1 c_2 \ldots c_{2n-1} c_{2n}}, \text{ where} \,\darkBlue{c_{2i-1} c_{2i}} =  01 \text{ if } b_i = 0 \text{ or } 10 \text{ if } b_i = 1.
\end{equation}

% where each pair $\bits{c_{2i-1} c_{2i}}$ implements the single bit $b_i$.
To understand this construction, note that elements of $L$ can only be modified using the intended swaps within a single pair (i.e., $\bits{c_{2i-1} c_{2i}}$) rather than between the pair~$\bits({c_{2i} c_{2i+1}})$.
Similar arguments for $\COMBTRANS$ show that the intended swaps are the only transpositions that can modify members of $L$.

\begin{figure}[h]
    \centering
    \begin{subfigure}{0.32\textwidth}
        \includegraphics[width=0.935\textwidth]{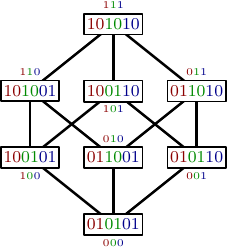}
        \caption{$\darkRed{c_1 c_2} \darkBlue{c_3 c_4} \darkGreen{c_5 c_6}$ from \eqref{eq:comboSwap} with swaps or transpositions.} % (i.e., adjacent-transpositions)
        \label{fig:combos_swap}
    \end{subfigure}
    \hfill
    \begin{subfigure}{0.32\textwidth}
        \includegraphics[width=\textwidth]{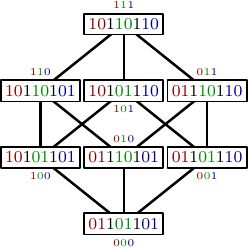}
        \caption{$\darkRed{c_1 c_2} \, \padding{1} \, \darkBlue{c_3 c_4} \, \padding{1} \, \darkGreen{c_5 c_6}$ from \eqref{eq:comboComplement} with substring complements.}
        \label{fig:combos_complement}
    \end{subfigure}
    \hfill
    \begin{subfigure}{0.32\textwidth}
        \includegraphics[width=\textwidth]{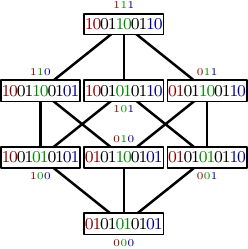}
        \caption{$\darkRed{c_1 c_2} \, \padding{01} \, \darkBlue{c_3 c_4} \, \padding{01} \, \darkGreen{c_5 c_6}$ from \eqref{eq:comboReverse} with substring reversals.}
        \label{fig:combos_reverse}
    \end{subfigure}
    \caption{
     Gray code reduction to prove the NP-hardness for combinations.
    }
    \label{fig:combos}
\end{figure}

The list in \eqref{eq:comboSwap} does not establish the result for $\COMBCOMP$.
This is because one substring complement can modify multiple pairs of bits e.g., $\overline{\bits{0 1} \, \bits{1 0}}$ gives $\bits{1 0} \, \bits{0 1}$.
To avoid this, we insert a padding bit~$\padding{1}$ between pairs and use these combinations from $\COMBOS{3n-1}{2n-1}$
\begin{equation} \label{eq:comboComplement}
    L := \{ \bits{c_1 c_2} \, \padding{1} \, \bits{c_3 c_4} \, \padding{1} \, \ldots \, \padding{1} \, \bits{c_{2n-1} c_{2n}} \; | \; \bits{c_{2i-1} c_{2i}} \in \{01,10\} \text{ for } 1 \leq i \leq n \, \}.
\end{equation} 
The $\padding{1}$ bits of padding prevent any substring complement of length $>2$ from modifying elements of $L$;
it is also clear that substring complements of length $1$ cannot modify $L$.
As a result, the only valid complements are internal to the pair ($\bits{c_{2i-1} c_{2i}}$).

Similarly, the list in \eqref{eq:comboComplement} does not establish the result for~$\COMBREV$.
This is because a single substring reversal can reverse multiple pairs of bits e.g., reversing $\bits{0 1} \, \padding{1} \, \bits{0 1}$ gives $\bits{1 0} \, \padding{1} \, \bits{1 0}$.
To avoid this, we use two bits of padding $\padding{01}$ and the following list of combinations from $\COMBOS{4n-1}{2n-1}$
\begin{equation} \label{eq:comboReverse}
    L := \{ \bits{c_1 c_2} \, \padding{01} \, \bits{c_3 c_4} \, \padding{01} \, \ldots \, \padding{01} \, \bits{c_{2n-1} c_{2n}} \; | \; \bits{c_{2i-1} c_{2i}} \in \{01,10\} \text{ for } 1 \leq i \leq n \, \}.
\end{equation} 
The~$\padding{01}$ padding ensures that the only substring reversals that can modify the elements of~$L$ are the intended swaps within the pair ($\bits{c_{2i-1} c_{2i}}$). 

Similar to equation~\ref{eq:comboSwapTranslate}, we can efficiently map a binary string ~$b_1 b_2\ldots b_n$ with combinations of length~$3n-1$ with~$2n-1$ many 1s and of length~$4n-1$ with~$3n-1$ many 1s, comprising of padding bits for \COMBCOMP and \COMBREV, respectively.
\qed
\end{proof}

\section{Problems on Permutations including Pattern-Avoidance}
\label{sec:permute}
For the set of permutations~$\PERMUTE{n}$, we consider an ordering where two consecutive permutations differ by \emph{swaps}. 
Gray codes then emerged involving \emph{transpositions} \cite{compton1993doubly}, \cite{shen2013hot}, as well as \emph{prefix-reversals} \cite{Ord67} and \emph{prefix shifts} \cite{corbett1992rotator} and ~\cite{sawada2019solving}.
Given a permutation $\pi = p_1 \cdots p_n$ with a substring $p_i \cdots p_j$ where $p_i > p_{i+1} \cdots  p_j$ , a \emph{right-jump} of the value $p_i$ by $j - i$ steps is a cyclic left rotation of this substring by one position to $p_{i+1} \cdots p_jp_i$. 
Analogously, we define a \emph{left-jump}. 
Jump Gray codes were given in ~\cite{hartung2022combinatorial}.
\begin{problem}
    \problemtitle{$\PERMSWAP / \PERMTRANS / \PERMREV/ 
    \PERMROT /
   \PERMJUMP$}
    \probleminput{A list $L$ of permutations from $\PERMUTE{n}$.}

    \problemquestion{Is there a swap / transposition / substring reversal/ substring rotation/ jump Gray code for $L$? 
    }
\end{problem}

\begin{theorem} \label{thm:permute}
$\PERMSWAP, \PERMTRANS, \PERMREV, \PERMROT, \PERMJUMP$ are NP-hard.
\end{theorem}

\begin{proof}
We use a Gray code reduction; see Figure \ref{fig:permute}.
Let~$b_1b_2 \cdots b_n$ is a bitstring of length~$n$.
For \textsc{PermSwapGC}, we use the following list of permutations from $\PERMUTE{2n}$
\begin{equation} \label{eq:permSwap}
    L := \{ \darkBlue{p_1 p_2 \ldots p_{2n-1} p_{2n}} \; | \; \darkBlue{p_{2i-1}, p_{2i}} \in \{2i-1,  2i\} \text{ for } 1 \leq i \leq n \, \},
\end{equation}
where each bit $b_i$ is implemented by the pair $\darkBlue{(p_{2i-1}, p_{2i})}$.
We map each binary string of length $n$ to a permutation of length $2n$,
\begin{equation} \label{eq:permSwaptranslate}
    f(b_1 \cdots b_n) := \darkBlue{p_1 \ldots p_{2n-1} p_{2n}}, \text{ where } \darkBlue{p_{2i-1} p_{2i}} = 2i{-}1 \, 2i \text{ if } b_i = 0 \text{ or } 2i \, 2i{-}1 \text{ if } b_i = 1.
\end{equation}
% The pair $\darkBlue{(p_{2i-1}, p_{2i})}$ implements the single bit $b_i$. 
The elements of $L$ can only be modified using the intended swaps within the pair~$\darkBlue{(p_{2i-1}, p_{2i})}$. 
Similarly, the intended swaps are the only transpositions, substring reversals,  substring rotations, and jumps that can modify the elements of $L$. Note that unlike Theorem~\ref{thm:combos}, we need not redefine the set~$L$ for different operations.
\qed
\end{proof}

%For journal version :
%For length~$n$ of a binary string, each permutation in the set~$\PERMUTE{n}$ in \eqref{eq:permSwap} is $2n$. When we encode each element of the permutation by its binary representation, the total number of bits required to encode one permutation is~$2n\ceil{log_2(2n)}$, which is polynomial in $n$.
%NP membership

\begin{figure}[h]
    \centering
    \begin{subfigure}{0.49\textwidth}
        \centering
        \includegraphics[scale=1.0]{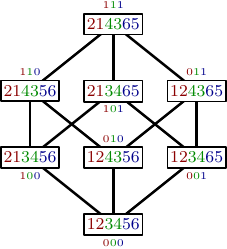}
        \caption{$\darkRed{p_1 p_2 \darkGreen{p_3 p_4} \darkBlue{p_5 p_6}}$ from \eqref{eq:permSwap} with swaps.}
        \label{fig:permute_swap}
    \end{subfigure}
    \hfill
    \begin{subfigure}{0.49\textwidth}   
        \centering
        \includegraphics[scale=1.0]{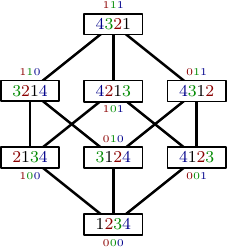}
        \caption{$p_1\darkRed{p_2} \darkGreen{p_3} \darkBlue{p_4}$ from jumps.}
        \label{fig:peakless-permute-jump}
    \end{subfigure}
    \caption{Gray code reductions to prove the  NP-hardness for permutations.}
    \label{fig:permute}
\end{figure}
A \emph{peak} in a permutation~$p_1 \cdots p_n$ is a triple $p_{i-1}p_ip_{i+1}$
with~$p_{i-1} < p_{i} > p_{i+1}$. A set of permutations without a peak, also called a \emph{peakless permutation} is denoted by ~$\PEAKLESSPERM{n}$.
We consider the following Gray coding problem on peakless permutations.

\begin{problem}    \problemtitle{{\textsc{PeakPermJumpGC}}}
    \probleminput{A list $L$ of permutations from $\PEAKLESSPERM{n}$.}
    \problemquestion{Is there a jump Gray code for $L$? 
    }
\end{problem}

\begin{theorem} \label{thm:peaklesspermute}\textsc{PeakPermJumpGC} is NP-hard.
\end{theorem}

\begin{proof}
We use a Gray code reduction; see Figure~\ref{fig:permute}.
We define the list of peakless permutations as $\PEAKLESSPERM{n} := \{\darkBlue{p_1p_2\cdots p_n} \mid \nexists i \, \text{ where } \, p_{i-1} < p_{i} > p_{i+1}\}$.

We map a binary string~$b_2 \cdots b_{n}$ to a permutation~$\pi \in \PEAKLESSPERM{n}$  as follows: we start with 1 and then insert the values~$i = 2, \ldots, n$ one by one, either at the leftmost or rightmost position, depending on whether the bit $b_i$ is 1 or 0, respectively. Thus a bitstring of length~$n$ maps to a permutation of length~$n+1$ and
two permutations in~$\PEAKLESSPERM{n}$ differ in a jump if and only if the mapped bitstrings differ in a bitflip. 
\qed
\end{proof}

For $n \geq k$, let $\pi \in \PERMUTE{n}$ and $\tau \in \PERMUTE{k}$. 
We say that $\pi $ \emph{contains} $\tau$, if and only if $\pi = p_1 \cdots p_n$ contains a subpermutation $p_{i_1} \cdots p_{i_k}$ with the same relative order as the elements in $\tau$. 
Otherwise, $p$ \emph{avoids}~$\tau$. 
We denote $\PERMUTE{n}(\tau)$ as the set of all permutations of length~$n$ that avoids $\tau$. 
Moreover, $\PERMUTE{n}(\tau_1 \wedge \cdots \wedge \tau_{\ell} )$ is the set of permutations of length~$n$ avoiding each of the patterns $\tau_1, \ldots , \tau_{\ell}$. Gray codes for pattern-avoiding permutations appeared in \cite{dukes2008combinatorial}, \cite{baril2009more}, \cite{hartung2022combinatorial}.

\begin{remark} \label{rem:peakless}
Peakless permutations of $[n]$ are (132 $\wedge$ 231)-avoiding permutations of~$[n]$. 
\end{remark}
We extend Theorem~\ref{thm:peaklesspermute} to multiple permutation patterns consisting of ANDs.
\begin{corollary}
\label{cor:pattern_or_and}
For a set of permutation patterns~$\{\tau_1, \tau_2, \ldots, \tau_{\ell}\}$, if every $\tau_i$ contains a peak, 
then the Gray code problem 
is NP-hard for on jumps the list of permutations from~${\PERMUTE{n}(\tau_1 \wedge \ldots \wedge \tau_{\ell})}$.
\end{corollary}

\begin{figure}[h]
 \centering
\begin{subfigure}{0.32\textwidth}
   \centering
    \includegraphics[width=\textwidth]{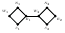}
    \caption{The diamond-path $D_2$.}
     \label{fig:diamonds_graphs}  
\end{subfigure}
\hfill
    \begin{subfigure}{0.32\textwidth}
       \centering      
       \includegraphics[width=\textwidth]{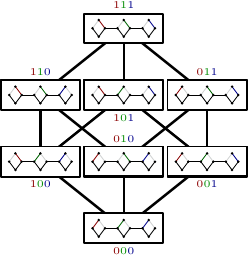}
        \caption{Spanning trees of $D_3$.}
        \label{fig:graphs1}
    \end{subfigure}
    \hfill
    \begin{subfigure}{0.32\textwidth}
        \centering
        \includegraphics[width=\textwidth]{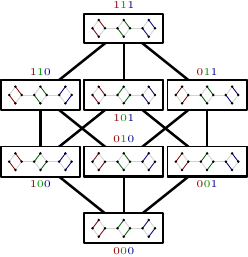}
        \caption{Perfect matchings of $D_3$.}
        \label{fig:graphs2}
    \end{subfigure}
    \caption{
    Gray code reductions for set partitions and graphs.
    }
    \label{fig:sp}
\end{figure}

%With this implementation, it is clear that every \emph{bitflip} corresponds to a \emph{move} of an element from one set to another in a set partition.

%Journal version extensions:
%The problem is in NP explain encoding.
%% Please use bibtex, 
\section{Problems on Graphs}
\label{sec:graphs}

% In this section we show hardness for Gray code problems on graphs. 

% \subsection{Diamond-path graph}
%\begin{wrapfigure}[8]{L}{.25\textwidth}
 %  \centering
  %   \includegraphics[width=.3\textwidth,page=1]{graphics/graphs.pdf}
   %  \caption{The diamond-path graph $D_2$.}
    % \label{fig:diamonds_graphs}
 %\end{wrapfigure}

In this section, we discuss some graph-related problems.
Our reductions are based on a particular graph namely, the diamond-path graph.
Informally, the $n$-th diamond-path graph consists of $n$ squares joined through edges. 
More formally, we define the graph $D_n$ by considering the vertex set $\{N,S,E,W\} \times [n]$, and the edge set $\{N_i E_i \cup N_i W_i  \cup  S_i E_i \cup S_i W_i \mid i \in [n]\} \cup   \{E_i W_{i+1} \mid i \in [n-1]\}$;
see Figure~\ref{fig:diamonds_graphs}.
%For $i \in [n]$, we call the set $\{(N,i), (E,i), (W,i), (S,i)\}$ the \emph{$i$th diamond}.  

\subsection{Spanning Trees}
\label{sec:graphs_spanning}
A \emph{spanning tree} of a graph $G$ is a connected acyclic subgraph of $G$. We denote the set of all spanning trees of a fixed graph $G$ by $\SPANNING{G}$.
We say that two spanning trees $T,T'$ of $G$ differ in an \emph{edge exchange} if they differ in exactly two edges; i.e., there exist edges $e \in T\setminus T'$ and $f\in T'\setminus T$ such that $T=T'+e-f$. Gray codes for spanning trees under edge exchanges have been widely studied from both the combinatorial and computational point of view~\cite{DBLP:journals/corr/abs-2304-08567,DBLP:conf/fun/MerinoMW22,MR762893,Smith}.

\begin{problem}    
\problemtitle{{$\STEDGEEXG$}}
    \probleminput{A list $L$ of spanning trees from~$\SPANNING{D_n}$.}
    \problemquestion{Is there an edge-exchange Gray code for~$L$? 
    }
\end{problem}
\begin{theorem} \label{thm:spanning1}
$\STEDGEEXG$ is NP-hard.
\end{theorem}
 
\begin{proof}[of Theorem~\ref{thm:spanning1}]
We use a Gray code reduction; see Figure \ref{fig:graphs1}.
For bitstrings in $\BINARY{n}$, we define the following list of spanning trees of $D_n$
\begin{equation*}
L := \{W_i S_i \mid i \in [n]\} \cup \{E_i S_i \mid i \in [n]\} \cup \{E_i W_{i+1} \mid i \in [n-1]
\cup \bigcup_{i : b_i = 0} \{W_i N_i\} \cup \bigcup_{i : b_i = 1} \{E_i N_i\}.
\end{equation*}
In other words, a spanning tree in~$L$ contains all the edges between diamonds, the edges that are incident on $S_i$ vertices, and to join $N_i$, we use either the edge $W_i$ or $E_i$, depending on the value of $b_i$.
Therefore, two spanning trees~$T_b, T_{b'} \in L$ differ in an edge exchange if and only if $b, b' \in \BINARY{n}$ differ in a bitflip.
\qed
\end{proof}

\subsection{Perfect Matchings}

A perfect matching of a graph is a set of edges $M \subseteq E$ such that every vertex is incident to exactly one edge in $M$.
We denote the set of all perfect matchings of a fixed graph $G$ by $\PMATCHING{G}$.
We say that two perfect matchings $M,M'$ of $G$ differ in an \emph{alternating cycle} if their symmetric difference forms a cycle in $G$. 
Every graph with a perfect matching has an alternating cycle Gray code for $\PMATCHING{G}$ that can be efficiently computed~\cite{DBLP:journals/corr/abs-2304-08567,MR762893}.

\begin{problem}    \problemtitle{{$\PMALTCYCLE$}}
    \probleminput{A list $L$ of perfect matchings from~$\PMATCHING{D_n}$.}
    \problemquestion{Is there an alternating cycle Gray code for~$L$? 
    }
\end{problem}
\begin{theorem} \label{thm:pmatching}
$\PMALTCYCLE$ is NP-hard.
\end{theorem}

% The flip graph $\flipGraph{\PMATCHING{G}}{\ALTCYCLE{G}}$ is exactly the skeleton of the following polytope~\cite{MR371732}:
% \[\mathop{conv}(\{\ind_M \in \{0,1\}^E  \mid \text{$M$ is a perfect matching of $G$}  \}) = \left\{ \sum_{M \in \PMATCHING{G}} \lambda_M\ind_M  \in \{0,1\}^E \mid \text{$\sum_{M\in \PMATCHING{G}}\lambda_M=1$.}  \right\}\]

\begin{proof}[of Theorem~\ref{thm:pmatching}]
We use Gray code reductions; see Figure \ref{fig:graphs2}.
For bitstrings in $\BINARY{n}$, we define the following list of perfect matchings of $D_n$
\begin{equation*}
L := \bigcup_{i : b_i = 0} \{W_i N_i, E_i S_i\} \cup \bigcup_{i : b_i = 1} \{E_i N_i, W_i S_i\}.
\end{equation*}
There are two possible choices of perfect matchings for every diamond in~$D_n$.
Therefore, two perfect matchings~$M_b, M_{b'} \in L$ differ in an alternating cycle if and only if $b, b' \in \BINARY{n}$ differ in a bitflip.
\qed
\end{proof}
Theorems~\ref{thm:spanning1} and~\ref{thm:pmatching} also extend the hardness when we are given the host graph $G$ as input and ask for edge-exchange or alternating cycles Gray codes for lists of spanning trees or perfect matchings, respectively, of $G$.

\section{Final Remarks}
\label{sec:final}
We proved that the Gray coding problems are NP-complete for various classical objects.
Future work could involve investigating optimization and approximation variants.
Furthermore, our techniques apply to many other objects, for example, those involving geometry, that we plan to explore in the full version of this paper.
% For example, the following questions are natural for a subset of $n$-bit binary strings.
% \begin{itemize}
%     \item What is the longest Gray code that can be constructed using a subset of the provided strings?
%     \item What is complexity of finding a Gray code that uses a given fraction of the provided strings?
% \end{itemize}

We note that we were unable to establish NP-hardness for certain subset problems using grid or hypercube reductions.
These problems include operations that do not support \emph{independent involutions}.
In other words, there is no way to make and unmake multiple local changes, which is a hallmark of hypercube reductions.
It is also important to note that some subset problems are poly-time solvable (e.g., those associated with de Bruijn sequences and universal cycles).

% Specific challenges were born from two of the most stubborn Gray code problems that were solved during the past decade:
% the middle levels~\cite{mutze2016proof} and sigma-tau problems~\cite{sawada2019solving}.
% In the first case, there is a global condition (i.e., the number of $1$s) that prevents involutions from acting independently.
% In the second case, every operation changes the first symbol, so the operations are again dependent.
% We hope to investigate these issues in future work.

\bibliography{refs}
\end{document}